\documentclass[sn-mathphys]{sn-jnl}

\usepackage{stmaryrd}
\usepackage{amssymb}
\usepackage{amsmath}
\usepackage{graphicx}
\usepackage{enumitem}
\usepackage{multicol}
\usepackage{proof}

\usepackage{program}
\catcode`\|=12\relax

\usepackage{tikz}
\usetikzlibrary{arrows, topaths, calc, positioning}
\tikzset{>=stealth}
\tikzstyle{node} = [circle, minimum size = 1.4mm, inner sep = 0mm, color=black, fill]
\tikzstyle{hyperedge} = [rectangle, minimum width = 5mm, minimum height = 5mm, draw, inner sep = 0mm]

\newcommand{\eqdef}{\mathrel{\mathop:}=}

\newcommand{\LC}{\mathrm{L}}
\newcommand{\LP}{\mathrm{LP}}
\newcommand{\LPG}{\mathrm{LPG}}
\newcommand{\lBAM}{\mathrm{lBAM}}

\newcommand{\Nat}{\mathbb{N}}

\jyear{2022}%

\theoremstyle{thmstyleone}%
\newtheorem{theorem}{Theorem}
\newtheorem{lemma}{Lemma}
\newtheorem{corollary}{Corollary}
\theoremstyle{thmstyletwo}%
\newtheorem{remark}{Remark}%

\theoremstyle{thmstylethree}%
\newtheorem{definition}{Definition}%
\newtheorem{construction}{Construction}%

\begin{document}

\title[Grammars over LP: Recognizing Power and Connection to BVASS]{Grammars over the Lambek Calculus with Permutation: Recognizing Power and Connection to Branching Vector Addition Systems with States}


\author*{\fnm{Tikhon} \sur{Pshenitsyn}}
\email{ptihon@yandex.ru}

\affil*{\orgdiv{Department of Mathematical Logic and Theory of Algorithms}, \orgname{Lomonosov Moscow State University}, \orgaddress{\street{GSP-1, Leninskie Gory}, \city{Moscow}, \postcode{119991}, \country{Russia}}}


\abstract{
In (Van Benthem, 1991) it is proved that all permutation closures of context-free languages can be generated by grammars over the Lambek calculus with the permutation rule ($\LP$-grammars); however, to our best knowledge, it is not established whether the converse holds or not. In this paper, we show that $\LP$-grammars are equivalent to linearly-restricted branching vector addition systems with states and with additional memory (shortly, lBVASSAM), which are modified branching vector addition systems with states. Then an example of such an lBVASSAM is presented, which generates a non-semilinear set of vectors; this yields that $\LP$-grammars generate more than permutation closures of context-free languages. Moreover, equivalence of $\LP$-grammars and lBVASSAM allows us to present a normal form for $\LP$-grammars and, as a consequence, prove that $\LP$-grammars are equivalent to $\LP$-grammars without product. Finally, we prove that the class of languages generated by $\LP$-grammars is closed under intersection.
}




\maketitle

\section{Introduction}\label{sec_introduction}

In the formal grammar theory, there are two families of approaches, which are, in some sense, opposed to each other: generative grammars and categorial grammars. Generative grammars are rule-based: a string belongs to a language generated by a grammar if and only if this string can be produced from a fixed start object of the grammar using rewriting rules specified in it. Categorial grammars work in an opposite way in the sense that a grammar does not produce a string step-by-step but rather takes a whole string in the first place and \textit{proves} that it belongs to its language. Hence to define any particular kind of categorial grammars we must specify what proof mechanism is used in its core.

One of prominent kinds of categorial grammars is Lambek categorial grammars. They are based on the Lambek calculus $\LC$, which is a logic designed to model syntax of natural languages \cite{Lambek58}. This is a substructural logic of the intuitionistic logic, namely, it is obtained from the latter by dropping structural rules of weakening, contraction, and permutation. In the Lambek calculus, types (i.e. formulas) are built from atomic ones using three operations: the left division $\backslash$, the right division $/$, and the product $\cdot$. Two divisions correspond to the implication in the intuitionistic logic; two variants of the implication arise since the order of types matters in the Lambek calculus. Following \cite{Lambek58} we consider \emph{sequents} as provable objects in $\LC$, which are structures of the form $A_1,\dotsc,A_n \to B$ where $n>0$ and $A_i,B$ are types. Finally, we define a Lambek grammar as a finite correspondence between symbols of an alphabet and types of $\LC$; besides, in a grammar, we must fix some \emph{distinguished type} $S$. Then a string $w = a_1\dotsc a_n$ belongs to the language generated by such a grammar if and only if we can replace each symbol $a_i$ by a type $T_i$ corresponding to it in such a way that the sequent $T_1,\dotsc,T_n \to S$ is derivable in $\LC$. 

One of the famous results proved by Pentus in \cite{Pentus93} is that Lambek grammars generate only context-free languages (the converse, i.e., that each context-free language without the empty word is generated by some Lambek grammar, was proved in \cite{Bar-Hillel60} in 1960). The proof uses several delicate tricks involving free-group interpretations and so-called binary-reduction lemma.

After the seminal work \cite{Lambek58} of Lambek, numerous modifications and extensions of $\LC$ have been introduced for different purposes. For each such modification one can define a corresponding class of categorial grammars in the same way as Lambek grammars are defined and then study what class of languages new grammars are able to generate. In particular, it is interesting to check if the ideas behind the theorem proved by Pentus in \cite{Pentus93} fit in a new class of categorial grammars. 

This work follows this agenda: we are going to investigate what languages can be generated by categorial grammars based on the Lambek calculus with the permutation rule (denoted as $\LP$). This calculus is obtained from $\LC$ by allowing one to freely change the order of types in left-hand sides of sequents. It is used and studied in, e.g., \cite{Benthem83,Benthem91}. In particular, in \cite{Benthem91}, it is proved that $\LP$-grammars generate all permutation closures of context-free languages; however, the converse statement was neither proved nor disproved in that paper. Stepan L. Kuznetsov \cite{Kuznetsov_personal} introduced this problem to me conjecturing that there is a counterexample to the converse statement. To my best knowledge, this question is still open; for instance, in \cite[p.~230]{Valentin12} the question of existence of a Pentus-like proof for $\LP$ is mentioned as an open one.

In this paper, we answer this question and show that $\LP$-grammars generate some languages that are not permutation closures of context-free languages (hence confirming Kuznetsov's conjecture). This is done by introducing an equivalent formalism called \emph{linearly-restricted branching vector addition systems with states and additional memory (lBVASSAM)}\footnote{We apologize to the reader for such long abbreviations.}. We prove that lBVASSAM generate exactly Parikh images of languages generated by $\LP$-grammars. The proof is inspired by the proof of the fact that double-pushout hypergraph grammars with a linear restriction on length of derivations can be embedded in hypergraph Lambek grammars presented in \cite{Pshenitsyn22}. In that work, we deal with a general formalism extending the Lambek calculus to hypergraph structures and investigate expressivity of the corresponding class of categorial grammars. Nicely, working with them on such a general level also provided us with methods applicable to $\LP$-grammars as well.

In this paper, we show how to transform an $\LP$-grammar into an equivalent lBVASSAM and vice versa. This also allows us to prove some nice facts about $\LP$-grammars: for example, we can prove that $\LP$-grammars are equivalent to $\LP(/)$-grammars, i.e. grammars based on the Lambek calculus with permutation and with division only. Another observation is that $\LP$-grammars are equivalent to $\LP$-grammars of order 2, i.e. to grammars with the maximal nesting depth of divisions being equal to 2.

This paper is organized as follows. In Section \ref{sec_preliminaries}, we introduce some preliminary notions and formaly define the Lambek calculus with permutation $\LP$ along with $\LP$-grammars. In Section \ref{sec_def_lbvassam}, we define linearly-restricted branching vector addition systems with states and additional memory. In Section \ref{sec_main_result}, we prove the equivalence theorem and some its corollaries. In Section \ref{sec_intersection}, we prove that languages generated by $\LP$-grammars are closed under finite intersection. In Section \ref{sec_conclusion}, we conclude. 

\section{Preliminaries}\label{sec_preliminaries}
Let us start with clarifying some notation used in the remainder of the paper. 

$\Sigma^\ast$ is the set of all strings over the alphabet $\Sigma$ (including the empty string $\varepsilon$). $\mathcal{M}(\Sigma)$ is the set of all finite nonempty multisets with elements from $\Sigma$. In this paper, we call a \emph{language} any subset of $\mathcal{M}(\Sigma)$. The length $|w|$ of a multiset $w$ is its cardinality; by $|w|_a$ we denote the number of occurrences of an element $a$ in $w$. The size $|v|$ of a vector $v = (v_1,\dotsc,v_k) \in \Nat^k$ equals $v_1+\dotsc+v_k$. By $e_i$ we denote the standard-basis vector $(0,\dotsc,0,1,0,\dotsc,0)$ where $1$ stands at the $i$-th position.

Hereinafter, given a multiset $\{a_1,\dotsc,a_n\}$, we often write $a_1,\dotsc,a_n$ instead, i.e. we omit braces. In the same spirit, when we write $\Gamma,\Delta$ for multisets $\Gamma$ and $\Delta$, this stands for their union $\Gamma \cup \Delta$.

If $\Sigma = \{a_1,\dotsc,a_k\}$ is a finite alphabet (with a fixed enumeration of symbols from $1$ up to $k$), then Parikh image of a multiset $w \in \mathcal{M}(\Sigma)$ is defined as $\pi(w) = (|w|_{a_1},\dotsc,|w|_{a_k})$. This definition is generalized to languages in an obvious way: $\pi(L) = \{\pi(w) \mid w \in L\}$ (for $L \subseteq \mathcal{M}(\Sigma)$). We can also consider the inverse Parikh image: $\pi^{-1}(V) = \{w \in \mathcal{M}(\Sigma) \mid \pi(w) \in V\}$. Clearly, it always holds that $\pi(\pi^{-1}(V)) = V$ and $\pi^{-1}(\pi(L))=L$.

\subsection{Lambek Calculus With Permutation}
In this section, we define the Lambek calculus with permutation $\LP$ in the Gentzen style. We fix a countable set of \emph{primitive types} $Pr$ and define the set of \emph{types} as follows: $Tp \eqdef Pr \mid Tp/Tp \mid Tp \cdot Tp$ (in contrast to the Lambek calculus without permutation, we do not need to introduce the left division $\backslash$ here). A \emph{sequent} is a structure of the form $A_1,\dotsc,A_n \to B$ where $n > 0$, and $A_i$, $B$ are types. The multiset $A_1,\dotsc,A_n$ is called an \emph{antecedent}, and $B$ is called a \emph{succedent}.

The only axiom of $\LP$ is $p \to p$ for $p \in Pr$. There are four inference rules:
$$
\infer[(/\to)]{\Gamma, B / A, \Pi \to C}{\Pi \to A & \Gamma, B \to C}
\qquad
\infer[(\to/)]{\Pi \to B / A}{\Pi, A \to B}
$$
$$
\infer[(\cdot\to)]{\Gamma, A \cdot B \to C}{\Gamma, A, B \to C}
\qquad
\infer[(\to\cdot)]{\Pi, \Psi \to A \cdot B}{\Pi \to A & \Psi \to B}
$$
Here small Latin letters $p,q,r,\dotsc$ represent primitive types; capital Latin letters $A,B,C,\dotsc$ represent types; capital Greek letters $\Gamma,\Delta,\dotsc$ represent finite multisets of types (besides, $\Pi,\Psi$ must be nonempty). By $\LP \vdash \Pi \to A$ we mean that the sequent $\Pi \to A$ is derivable in $\LP$. 

We call the type $B/A$ in the rules $(/\to)$ and $(\to/)$ \emph{major} as well as the type $A \cdot B$ in the rules $(\cdot\to)$ and $(\to\cdot)$.
\begin{definition}
	$A \times k$ is a shorthand notation for $\underbrace{A,\dotsc, A}_{k\mbox{ times}}$, and $A^k \eqdef \underbrace{A\cdot\dotsc\cdot A}_{k\mbox{ times}}$.
\end{definition}

The Lambek calculus with permutation can be restricted to the calculus without the product $\cdot$, i.e. we can consider a fragment of $\LP$ with division only. We will denote this fragment as $\LP(/)$.
\begin{definition}
	The \emph{length} of types and sequents is defined as follows:
	\begin{enumerate}
		\item $|p| = 1$;
		\item $|A \circ B| = |A|+|B|+1$, $\circ \in \{\cdot, /\}$;
		\item $|A_1,\dotsc,A_n \to B| = |A_1|+\dotsc+|A_n|+|B|$.
	\end{enumerate}
	
	The \emph{depth} of a type $A$ without products is defined as follows:
	\begin{enumerate}
		\item $d(p) = 0$, $p \in Pr$;
		\item $d(A/B) = \max\{d(A),d(B)+1\}$.
	\end{enumerate}
\end{definition}

The cut rule is admissible in $\LP$ (i.e. everything that can be derived using it can also be derived without it):
$$
\infer[(\mathrm{cut})]{\Gamma,\Pi \to B}{\Pi \to A & \Gamma,A \to B}
$$
This implies that the following rules are also admissible:
$$
\infer[(\cdot \to)^{-1}]{\Gamma,A,B \to C}{\Gamma,A\cdot B \to C}
\qquad
\infer[(\to /)^{-1}]{\Pi,A \to B}{\Pi \to B/A}
$$
Indeed, the first rule is in fact the application of the cut rule to the sequents $A,B \to A\cdot B$ and $\Gamma,A\cdot B \to C$; the second rule is the application of the cut rule to the sequents $\Pi \to B/A$ and $B/A,A \to B$. The above two rules are opposite to the rules $(\cdot \to)$ and $(\to /)$; in other words, we proved that the latter ones are invertible.

Now let us formulate the definition of $\LP$-grammars.
\begin{definition}\label{def_lp-grammar}
	An $\LP$-grammar is a tuple $G = \langle \Sigma,S,\triangleright \rangle$ where $\Sigma$ is a finite \emph{alphabet}, $S \in Tp$ is a \emph{distinguished type}, and $\triangleright \subseteq \Sigma \times Tp$ is a finite binary relation between symbols of the alphabet and types (in other words, one assigns several types to each element of $\Sigma$). 
	\\
	\emph{The language $L(G)$ generated by $G$} is the set of multisets $a_1,\dotsc,a_n \in \mathcal{M}(\Sigma)$ such that there exist types $T_1,\dotsc,T_n$ of $\LP$, for which it holds that:
	\begin{enumerate}
		\item $a_i \triangleright T_i$ ($i=1,\dotsc,n$);
		\item $\LP\vdash T_1,\dotsc, T_n \to S$.
	\end{enumerate}
\end{definition}

\begin{definition}
	We denote by $Tp(G)$ the set of all types invloved in $G$ (including $S$); more formally, $Tp(G) \eqdef \{T \mid \exists a: a \triangleright T\} \cup \{S\}$. Let us also inductively define the set $STp^+(G)$ of \emph{positive subtypes of $G$} and the set $STp^-(G)$ of \emph{negative subtypes of $G$} as follows:
	\begin{itemize}
		\item If there exists $a$ such that $a \triangleright T$, then $T \in STp^-(G)$;
		\item $S \in STp^+(G)$;
		\item If $A/B \in STp^\pm(G)$, then $A \in STp^\pm(G)$ and $B \in STp^\mp(G)$ (here $\pm$ and $\mp$ are either $+$ and $-$ resp. or $-$ and $+$ resp.);
		\item If $A\cdot B \in STp^\pm(G)$, then both $A$ and $B$ are in $STp^\pm(G)$. 
	\end{itemize}
	The set $STp(G)$ of all subtypes of $G$ is simply the union $STp^+(G) \cup STp^-(G)$. 
\end{definition}
It is clear that whenever we consider a derivation of a sequent as in Definition \ref{def_lp-grammar}, negative subtypes of $G$ may appear within it (not as proper subtypes) only in antecedents of sequents while positive subtypes may appear only in their succedents (proof is by the induction on the length of a derivation).

\section{lBVASSAM}\label{sec_def_lbvassam}
In this section, we aim to define a new formalism called linearly-restricted branching vector addition system with states and additional memory. This kind of systems is based on branching vector addition systems with states (BVASS) defined in \cite{VermaG05}, which in turn extend vector addition systems. The latter systems are introduced in \cite{KarpM69}; they represent a very natural and simple formalism, which is equivalent to well-known Petri nets. Countless modifications of vector addition systems are considered in the literature: VASP, VASS, AVASS, BVASS, EBVASS, EVASS, PVASS etc. These extensions are used for different purposes; it should be noted that one of them is proving undecidability of linear logic and its fragments \cite{LincolnMSS92, Kanovich95}. Speaking of branching vector addition systems with states \cite{VermaG05} (BVASS), they are developed as a natural extension of both vector addition systems and Parikh images of context-free grammars. Here is the formal definition of BVASS:
\begin{definition}\label{def_bvass}
	A \emph{branching vector addition system with states (BVASS)} is a tuple $G = \langle Q, \mathcal{P}_0,\mathcal{P}_1,\mathcal{P}_2,s,K\rangle$ where 
	\begin{enumerate}
		\item $K \in \Nat$ is the \emph{dimension of $G$};
		\item $Q$ is a finite set of \emph{states};
		\item $\mathcal{P}_0$ is a finite set of \emph{axioms} of the form $q(\nu)$ where $q \in Q$, $\nu \in \Nat^K$;
		\item $\mathcal{P}_1$ is a finite set of \emph{unary rules} of the form $p(x+\delta_2) \leftarrow q(x+\delta_1)$ where $p,q \in Q$, $\delta_1,\delta_2 \in \Nat^K$;
		\item $\mathcal{P}_2$ is a finite set of \emph{binary rules} of the form $p(x+y) \leftarrow q(x),r(y)$ where $p,q,r \in Q$;
		\item $s \in Q$ is the distinguished \emph{accepting state}.
	\end{enumerate}
	A formula $p(v)$ (let us call it a \emph{fact}) for $p\in Q$, $v \in \Nat^K$ is \emph{derivable in such a BVASS $G$} if one of the following holds (an inductive definition):
	\begin{itemize}
		\item $p(v) \in \mathcal{P}_0$;
		\item $v = w+\delta_2$ where $w \in \Nat^K$, $p(x+\delta_2)\leftarrow q(x+\delta_1) \in \mathcal{P}_1$, and $q(w+\delta_1)$ is derivable in $G$;
		\item $v = u+w$ where $u,w \in \Nat^K$, $p(x+y) \leftarrow q(x),r(y) \in \mathcal{P}_2$, and $q(u)$, $r(w)$ are derivable in $G$.
	\end{itemize}
	\emph{The language $L(G)$ generated by such a BVASS $G$} consists of vectors $v \in \Nat^K$ such that $s(v)$ is derivable in $G$.
\end{definition}
\begin{definition}
	\emph{The size of a derivation} is the total number of axioms and rule applications occurring in it.
\end{definition}

To introduce lBVASSAM, we make two changes in the definition of BVASS:
\begin{enumerate}
	\item We allow one to use ``additional memory''; this means that main derivable objects in such systems are still vectors $v \in \Nat^K$ but in the end of a derivation, i.e. when we obtain a fact $s(v)$, we cut off a part of $v$ (in other words, we project this vector onto a subspace $\Nat^k$ for some $k \le K$). Moreover, we require that at the end of a derivation there are only zeroes in the part of $v$ we cut off; hence we have a zero test, which, however, can be used only at the last step of a derivation.
	\item We limit the size of a derivation of a vector $v$ in BVASS by requiring that the size must not exceed $C |v|$ for some constant $C$ (we call this \emph{a linear restriction}). That is, the size of a derivation is bounded by a linear function of the size of the resulting vector. Note that this restriction immediately makes the reachability problem decidable and even places it in NP since in order to check that a vector $v$ is derivable, we only need to check all possible derivations of $s(v)$ of the length $\le C |v|$.
\end{enumerate}
Let us formally introduce both modifications. Given a vector $v \in \Nat^k$ and $K \ge k$, we denote by $\iota_K(v)$ the vector $v^\prime \in \Nat^K$ such that $v^\prime_i = v_i$ ($i=1,\dotsc,k$) and $v^\prime_i = 0$ for $i > k$.
\begin{definition}\label{def_bvassam}
	A \emph{branching vector addition system with states and additional memory (BVASSAM)} is a tuple $G = \langle Q, \mathcal{P}_0,\mathcal{P}_1,\mathcal{P}_2,s,k, K\rangle$ where all the components except for $k$ are defined as in Definition \ref{def_bvass}, and $0 \le k \le K$.
	\\
	\emph{The language $L(G)$ generated by such a BVASSAM $G$} consists of vectors $v \in \Nat^k$ such that $s(\iota_K(v))$ is derivable in $\langle Q, \mathcal{P}_0,\mathcal{P}_1,\mathcal{P}_2,s,K\rangle$.
\end{definition}

\begin{definition}\label{def_lbvassam}
	A \emph{linearly-restricted BVASSAM $G\vert_{C}$ (lBVASSAM)} is a BVASSAM $G$ equipped with a natural number $C$. 
	\\
	\emph{The language $L(G\vert_C)$ generated by this grammar} consists of vectors $u  \in \Nat^k$ such that there exists a derivation of $s(\iota_K(u))$ in $G$ of the size not greater than $C \cdot |u|$.
\end{definition}

\begin{remark}
	It is not so common to directly restrict the length of derivations for formal grammars; usually, in order to decrease expressive power of some class of formal grammars, restrictions are imposed on the structure of rules used in them. For example, unrestricted grammars may include rules of the form $\alpha \to \beta$ where $\alpha$ and $\beta$ are arbitrary strings; however, such grammars are too powerful. However, if we allow one to use only rules of the form $A \to \beta$, then we obtain context-free grammars, which are effectively decidable. 
	
	One example of imposing a linear restriction similar to that from Definition \ref{def_lbvassam} can be found in \cite{Rambow94} for multiset-valued linear index grammars; however, we failed to find a simple connection between these grammars and $\LP$-grammars, so we develop our own definitions. For us, the main motivation for such a restriction is Construction \ref{construction_lbvassam_to_lp} presented in the next section. It would also be interesting to consider other kinds of restrictions, e.g. a quadratic one: we might define a \emph{qBVASSAM} $G\vert^2_C$, i.e. a BVASSAM $G$ equipped with a positive integer $C$, for which we consider derivations of the size not greater than $C\cdot |u|^2$. In general, the restriction might be of the form $C\cdot f(u)$ for some function $f$.
	
	Note, however, that for each BVASSAM $G$ there exists a constant $C_G$ such that, if $s(v)$ is derivable in $G$, then the size of its derivation must be at least $C_G|v|$. Indeed, if $D_G$ is the maximum of all $|\delta_1|$ for $p(x+\delta_2) \leftarrow q(x+\delta_1) \in \mathcal{P}_1$, and of all $|\nu|$ for $p(\nu) \in \mathcal{P}_0$, then a derivation of the size $n$ always results in a fact $q(v)$ such that $|v|\le D_G \cdot n$ (the proof is by induction on $n$); finally, let $C_G$ equal $1/D_G$. To conclude this remark, the linear restriction is the least possible one that does not make the resulting language finite; this property makes it interesting by itself.
\end{remark}

\section{Equivalence of linearly-restricted BVASSAM and LP-grammars}\label{sec_main_result}

The main result of this paper is that lBVASSAM are equivalent to $\LP$-grammars in the sense that a language $L$ is generated by an $\LP$-grammar if and only if its Parikh image $\pi(L)$ is generated by an lBVASSAM. The proof is split into two directions: the ``if'' direction is proved in Section \ref{ssec_lbvassam_to_lp} and the ``only if'' one is proved in Section \ref{ssec_lp_to_lbvassam}.

\subsection{From lBVASSAM to LP-Grammars}\label{ssec_lbvassam_to_lp}
\begin{construction}[an $\LP$-grammar corresponding to a lBVASSAM]\label{construction_lbvassam_to_lp}
Given an lBVASSAM $G = G\vert_{C} = \langle Q, \mathcal{P}_0,\mathcal{P}_1,\mathcal{P}_2,s,k,K\rangle \vert_{C}$, our aim is to construct an $\LP$-grammar generating $\pi^{-1}(L(G))$ (we assume that it is over the alphabet $\Sigma = \{a_1,\dotsc,a_k\}$). Let us introduce several preliminary constructions:
\begin{enumerate}
	\item We consider all states from $Q$ as primitive types; besides, we introduce new primitive types $g_1,\dotsc,g_K,f$ ($g_i$ correspond to standard-basis vectors $e_i$ in $\Nat^K$, and $f$ stands for ``finish!''). 
	\\
	Let us agree on the following notation: if $v \in \Nat^K$ is a natural-valued vector, then $g^v \eqdef g_1^{v_1} \cdot \dotsc \cdot g_K^{v_K}$ and $g \times v \eqdef g_1 \times {v_1},\dotsc,g_K \times {v_K}$. Hence $g^v$ consists of $g_i$-s combined using the product $\cdot$ while $g \times v$ is a multiset consisting of $g_i$-s.
	\item For each $\varphi_0 = q(\nu) \in \mathcal{P}_0$ we define a type $T(\varphi_0) \eqdef f/ g^\nu/q$;
	\item For each $\varphi_1 = q(x+\delta_2) \leftarrow p(x+\delta_1) \in \mathcal{P}_1$ we define a type $T(\varphi_1) \eqdef (p \cdot g^{\delta_1})/g^{\delta_2}/q$;
	\item For each $\varphi_2 = q(x+y) \leftarrow p(x),r(y) \in \mathcal{P}_2$ we define a type $T(\varphi_2) \eqdef f/(f/r)/(f/p)/q$.
	\item Let $\mathcal{P}$ denote $\mathcal{P}_0 \cup \mathcal{P}_1 \cup \mathcal{P}_2$.
\end{enumerate}
	
Now we define an $\LP$-grammar itself, which we denote as $\LPG(G)$: $\LPG(G) \eqdef \langle \Sigma,f/s,\triangleright \rangle$ where $a_i \triangleright A$ iff $A = g_i \cdot T(\varphi_1)\cdot \dotsc \cdot T(\varphi_{j})$ for some $0 \le j \le C$ such that all $\varphi_l$ are from $\mathcal{P}$ ($i$ here changes from $1$ to $k$).
\end{construction}

The main idea of the construction is that we encode each axiom and rule $\varphi$ of an lBVASSAM by a type $T(\varphi)$ and then ``attach'' these types to a primitive type $g_i$.

\begin{remark}
	Formally, the definition of $A^l$ is not correct if $l=0$ (what is the product of $A$ zero times?); consequently, $g^v$ can also be undefined. A canonical way of understanding $A^0$ is $A^0 = \mathbf{1}$; however, we do not have the unit in $\LP$. Nevertheless, let us notice that in Construction \ref{construction_lbvassam_to_lp} types of the form $g^v$ appear only in certain positions, namely, within types of the form $f/g^\nu/q$ and $(p \cdot g^{\delta_1})/g^{\delta_2}/q$. This suggests the following treatment of the problematic cases:
	\begin{itemize}
		\item If $v_{i_1}\ne 0$, \dots, $v_{i_j} \ne 0$ ($j \ge 1$, $i_1<\dotsc<i_j$) and $v_t = 0$ for $t \not\in \{i_1,\dotsc,i_j\}$, then we define $g^v$ as $g_{i_1}^{v_{i_1}} \cdot \dotsc \cdot g_{i_j}^{v_{i_j}} $.
		\item If $\nu = \vec{0}$ in the type $f/g^\nu/q$, then we simply take a type of the form $f/q$ instead.
		\item If $\delta_1 = \vec{0}$ in the type $(p \cdot g^{\delta_1})/g^{\delta_2}/q$, then we take a type of the form $p/g^{\delta_2}/q$ instead.
		\item If $\delta_2 = \vec{0}$ in the type $(p \cdot g^{\delta_1})/g^{\delta_2}/q$, then we take a type of the form $(p\cdot g^{\delta_1})/q$ instead. Finally, if both $\delta_1 = \delta_2 = \vec{0}$, then we use the type $p/q$.
	\end{itemize}
\end{remark}

\begin{remark}\label{remark_polynomial}
	The total number $N$ of types involved in $\LPG(G)$ (more precisely, the cardinality of the binary relation $\triangleright$) can be estimated as follows: $K|\mathcal{P}|^C \le N \le K\left(|\mathcal{P}|+1\right)^C$; hence the transformation procedure is exponential w.r.t. $C$, but polynomial, if $C$ is fixed.
\end{remark}

Our goal is to prove the following theorem:
\begin{theorem}\label{th_bvass_to_lp}
	Let $G = G^\prime \vert_C$ be an lBVASSAM. Then $L(G) = \pi(L(\LPG(G)))$.
\end{theorem}

Let us prove several lemmas first.
\begin{lemma}\label{lemma_wolf_q}
	If $\LP \vdash B_1,\dotsc,B_m \to q$ where $B_i$ are from $STp(\LPG(G))$ and $q \in Q$, then $n=1$ and $B_1 = q$.
\end{lemma}
\begin{proof}
	The proof is by induction on the length of a derivation. The base case is trivial ($q \to q$ is an axiom). 
	
	To prove the induction step consider the last rule applied in a derivation of $B_1,\dotsc,B_m \to q$. 
	
	\textit{Case $(/\to)$}: 
	$$
	\infer[]{\Pi,q/D \to q}{q \to q & \Pi \to D}
	$$
	Here we apply the induction hypothesis to conclude that one of the premises is $q \to q$. Therefore, one of $B_i$-s must be of the form $q/D$; however, there are no types of such form in $STp(\LPG(G))$, which leads us to a contradiction. Consequently, it is not the case that $(/\to)$ is applied.
	
	\textit{Case $(\cdot \to)$}: 
	$$
	\infer[]{\Gamma, A\cdot B \to q}{\Gamma, A,B \to q}
	$$
	By the induction hypothesis, $\Gamma,A,B=q$; however, there are at least two types in $\Gamma,A,B$. This is a contradiction.
\end{proof}
We also need a similar lemma for sequents with elements of the form $g_i$ in the succedent. To formulate it, let us firstly introduce a function $\mathcal{D}$, which takes a multiset of types of $\LP$ and returns another multiset of types: $\mathcal{D}(p) = p$ ($p \in Pr$), $\mathcal{D}(A/B) = A/B$, $\mathcal{D}(A \cdot B) = \mathcal{D}(A), \mathcal{D}(B)$, $\mathcal{D}(\Gamma,\Delta) = \mathcal{D}(\Gamma), \mathcal{D}(\Delta)$. Informally, we replace all outermost products $\cdot$ in a multiset of types by commas.
\begin{lemma}\label{lemma_wolf_e_general}
	If $\LP \vdash B_1,\dotsc,B_m \to g^u$ where $B_i$ are from $STp(\LPG(G))$ and $\vec{0}\ne u \in \Nat^K$, then $\mathcal{D}(B_1,\dotsc,B_m)$ and $g \times u$ are equal as multisets.
\end{lemma}
\begin{proof}
	The proof is by induction on the length of a derivation of the sequent. The base case is trivial since in such a case we have a sequent of the form $g_i \to g_i$ for some $i$. To prove the induction step let us consider the last rule applied. Again, there are several cases:
	
	Case $(\to/)$ is impossible since there are no divisions in $g^u$.
	
	Case $(/\to)$: then the last rule application has the form
	$$
	\infer[(/\to)]{\Gamma,\Delta,E/B \to g^u}
	{
		\Gamma, E \to g^u &
		\Delta \to B
	}
	$$
	Applying the induction hypothesis we conclude that $E$ must be a product of several primitive types of the form $g_i$. However, there are no types of the form $E/B$ with such $E$ in $STp(\LPG(G))$ (indeed, for each type $C/D$ from $STp(\LPG(G))$ it is the case that $C$ includes either $f$ or some $q \in Q$). This allows us to draw a conclusion that the last rule application cannot be of the form $(/\to)$.
	
	For the case $(\cdot \to)$ the last rule application is of the form
	$$
	\infer[(\cdot\to)]{\Gamma,E_1\cdot E_2 \to g^u}
	{
		\Gamma,E_1, E_2 \to g^u
	}
	$$
	It suffices to notice that the function $\mathcal{D}$ returns the same output for both antecedents; the induction hypothesis completes the proof for this case.
	
	Case $(\to\cdot)$: the last rule application must be of the form
	$$
	\infer[]{\Gamma_1,\Gamma_2 \to g^u}
	{
		\Gamma_1 \to g^{u_1} &
		\Gamma_2 \to g^{u_2}
	}
	$$
	where $u_1+u_2 = u$. Applying the induction hypothesis to the premises, we immediately succeed.
\end{proof}

The following lemma is crucial:
\begin{lemma}\label{lemma_main_bvass_to_lp}
	A sequent of the form $g\times u, T(\psi_1),\dotsc, T(\psi_m), t \to f$ ($u \in \Nat^K$, $m \in \Nat$, $\psi_i \in \mathcal{P}$, $t \in Q$) is derivable if and only if $t(u)$ has a derivation in $G$ such that all the occurrences of axioms and rules in it are exactly those of $\psi_1,\dotsc,\psi_m$.
\end{lemma}
\begin{proof}
	There are two directions in this lemma, and we are going to prove both of them by a straightforward induction on the length of a derivation in a corresponding formalism. 
	
	Let us start with ``only if'', i.e. that provability of $g\times u, T(\psi_1),\dotsc, T(\psi_m), t \to f$ in $\LP$ implies derivability of $t(u)$ using $\psi_1,\dotsc,\psi_m$. The base case of the induction proof is trivial since the sequent of interest cannot be derived if $n=0$ since it is not an axiom (we have $t$ in the antecedent and $f \ne t$ in the succedent), and it cannot be obtained by applying any rule (because it includes only primitive types). 
	
	We proceed with the induction step by considering the last rule application in the derivation. Let $T(\psi_1)$ be the type that appears as the result of this application. There are three cases:
	
	\textbf{Case 1:} $\psi_1 = q(\nu) \in \mathcal{P}_0$. Then the final steps of the derivation must be of the form
	$$
	\infer[(/\to)]{
		g \times u, f/g^{\nu}/q,T(\psi_2),\dotsc, T(\psi_m),t \to f
	}{
		\infer[]{
			g \times u, f/g^{\nu},T(\psi_2),\dotsc, T(\psi_m) \to f
		}{
			g \times u_2, f,\Theta_2 \to f
			&
			g \times u_1, \Theta_1 \to g^\nu
		}
		&
		q \to q
	}
	$$
	Here $g \times u_1,g \times u_2 = g\times u$ (i.e. $u_1+u_2 = u$) and $\Theta_1, \Theta_2 = T(\psi_2),\dotsc, T(\psi_m)$ (as multisets, i.e., up to the order of types). Indeed, since $T(\psi_1) = f/g^\nu/q$ appears after the last rule application, there has to be a premise with $q$ in the succedent; according to Lemma \ref{lemma_wolf_q} this premise must be of the form $q \to q$. Since the only type from $Q$ in the antecedent of the sequent $g \times u, f/g^{\nu}/q,T(\psi_2),\dotsc, T(\psi_m),t \to f$ is $t$, we conclude that $q=t$.
	
	The sequent $g \times u, f/g^{\nu},T(\psi_2),\dotsc, T(\psi_m) \to f$ can be obtained only by an application of $(/\to)$. However, any of $T(\psi_i)$ ($i\in \{2,\dotsc,m\}$) cannot be the major type of this rule because each of them is of the form $A/r$ for some $r \in Q$ and for some $A$; if it was major, then it would be the case that $r$ is present in the antecedent of $g \times u, f/g^{\nu},T(\psi_2),\dotsc, T(\psi_m) \to f$, which does not hold. Consequently, $f/g^\nu$ must be major, thus the next step of the derivation is as shown above.
	
	Now, let us examine the sequent $g \times u_2, f,\Theta_2 \to f$. We claim that it must be an axiom. Otherwise, it appears as the result of some rule application; more precisely, it must be an instance of $(/\to)$. However, we notice that all the types with division in this sequent are of the form $A/r$ for some $r \in Q$; if the last rule application was $(/\to)$, then it would be of the form (according to Lemma \ref{lemma_wolf_q}):
	$$
	\infer[]{
		g \times u_2, f,\Theta_2 \to f
	}{
		\dotsc \to f
		&
		r \to r
	}
	$$ 
	This would imply that $r$ is present in $g \times u_2,f,\Theta_2$, which is not the case. Concluding we obtain that $g \times u_2$ and $\Theta_2$ are empty (in what follows that $u_2 = \vec{0}$).
	
	Now it suffices to apply Lemma \ref{lemma_wolf_e_general} to $g \times u_1, \Theta_1 \to g^\nu$, which allows us to draw a conclusion that $g \times u_1 = g\times \nu$ and that $\Theta_1$ is empty. Finally, we have $u = u_1 = \nu$, $q=t$ and hence $u$ is derivable in $G$ using $\psi_1$: $t(u) = q(\nu) = \psi_1$ is an axiom. 
	
	\textbf{Case 2:} $\psi_1 = q(x+\delta_2) \leftarrow p(x+\delta_1) \in \mathcal{P}_1$. Then the last steps of the derivation must be of the form
	$$
	\infer[(/\to)]{
		g\times u, (p\cdot g^{\delta_1})/g^{\delta_2}/q,T(\psi_2),\dotsc,T(\psi_m), t \to f
	}{
		\infer[(/\to)]{
			g\times u, (p\cdot g^{\delta_1})/g^{\delta_2},T(\psi_2),\dotsc,T(\psi_m) \to f
		}{
			\infer[(\cdot \to)]{
				g \times u_1, p\cdot g^{\delta_1},\Theta_1 \to f
			}{
				g \times u_1, g\times \delta_1, \Theta_1, p \to f
			}
			&
			g \times u_2, \Theta_2 \to g^{\delta_2}
		}
		&
		q \to q
	}
	$$
	Here $g \times u_1 + g \times u_2 = g \times u$ (equivalently, $u_1 + u_2 = u$) and $T(\psi_2),\dotsc, T(\psi_m) = \Theta_1, \Theta_2$. Reasonings concerning the last two steps of the derivation are the same as for the previous case; they imply that $q=t$, $u_2 = \delta_2$, and that $\Theta_2$ is empty. Besides, we force using the rule $(\cdot \to)$, which is possible due to admissibility of the rule $(\cdot \to)^{-1}$; in fact, in the above derivation, $(\cdot\to)$ is used multiple times in a row, since we eliminate all the products in the type $p \cdot g^{\delta_1}$ and obtain a multiset of types $g\times \delta_1, p$. 
	\\
	Finally, we apply the induction hypothesis to the sequent $g \times u_1, g \times {\delta_1}, \Theta_1, p \to f$ and obtain that $p(u_1 + \delta_1)$ has a derivation in $G$ such that axioms and rules used in it are exactly those from $\Theta_1 = \psi_2,\dotsc,\psi_m$. It remains to notice that $g\times u = g \times (u_1+\delta_2)$ and $t(u) = q(u_1+\delta_2)$ can be derived from $p(u_1+\delta_1)$ using $\psi_1$. This concludes the proof.
	
	\textbf{Case 3.} $\psi_1 = q(x+y) \leftarrow p(x),r(y)$. Then the last steps of the derivation must be of the form
	$$
	\infer[(/\to)_1]{
		g \times u, f/(f/r)/(f/p)/q,T(\psi_2),\dotsc,T(\psi_m), t \to f
	}{
		\infer[(/\to)_2]{
			g \times u, f/(f/r)/(f/p),T(\psi_2),\dotsc,T(\psi_m) \to f
		}{
			\infer[(/\to)_3]
			{
				g \times u_2,g \times u_3, f/(f/r),\Theta_2,\Theta_3 \to f
			}{
				g \times u_3, f,\Theta_3 \to f
				&
				\infer[(\to /)]{g \times u_2, \Theta_2 \to f/r}{g \times u_2, \Theta_2, r \to f}
			}
			&
			\infer[(\to /)]{g \times u_1, \Theta_1 \to f/p}{g \times u_1, \Theta_1, p \to f}
		}
		&
		q \to q
	}
	$$
	Here $q=t$, $g \times u = g \times u_1 + g \times u_2 + g \times u_3$ (i.e. $u_1+u_2+u_3 = u$) and $\Theta_1,\Theta_2,\Theta_3 = T(\psi_2),\dotsc,T(\psi_m)$; the applications of $(/\to)$ are numbered in order to refer to them. Reasonings for this case are similar to those for Cases 1 and 2. The main observation is that only the type $f/(f/r)/(f/p)$ (the type $f/(f/r)$) can be major in a rule application with number $2$ (number $3$ resp.) because all other types in the antecedent are either primitive or of the form $A/x$ where $x \in Q$; however, there is no primitive type $x \in Q$ in the antecedent of each corresponding sequent. The rule applications of $(\to /)$ are forced using invertibility of $(\to /)$. The sequent $g \times u_3, f,\Theta_3 \to f$ can be derivable only if it is axiom, i.e. if $g\times u_3$, $\Theta_3$ are empty. Finally, we apply the induction hypothesis to $g \times u_2,\Theta_2, r \to f$ and $g \times u_1,\Theta_1, p \to f$ and conclude that $r(u_2)$ and $p(u_1)$ can be derived in $G$ with the multiset of axioms and rules used in total in both derivations equal to $\psi_2,\dotsc,\psi_m$. Finally we apply the rule $\psi_1$ and come up with $q(u_1+u_2) = q(u)$ as desired.
	
	Speaking of the ``if'' direction, we need to transform a derivation of $t(u)$ in $G$ into a derivation of $g \times u, T(\psi_1),\dotsc, T(\psi_m), t \to f$ where $\psi_1,\dotsc,\psi_m$ are all the axiom and rule occurrences in the derivation of $t(u)$. This is done straightforwardly by induction on $m$; in fact, each axiom or rule application in $G$ is remodeled in $\LP$ in the same way as shown above (Cases 1-3).
\end{proof}

We are ready to prove Theorem \ref{th_bvass_to_lp}.
\begin{proof}[Proof (of Theorem \ref{th_bvass_to_lp})]
	Let $w = w_1,\dotsc,w_m$ belong to $L(\LPG(G))$ (where $w_i \in \Sigma$). This is the case if and only if $\LP \vdash A_1,\dotsc, A_{m} \to f/s$ where $w_i \triangleright A_i$ ($i=1,\dotsc,m$) for some types $A_i$. Each type $A_i$ in the antecedent of this sequent is either primitive or is a product of several types (cf. Construction \ref{construction_lbvassam_to_lp}). 
	
	We claim that $w \in L(\LPG(G))$ if and only if the following sequent is derivable for some $n \le C |w| = C m$ and some $\varphi_i \in \mathcal{P}$:
	\begin{equation}\label{eq_bvass_to_lp}
		g\times v, T(\varphi_1),\dotsc, T(\varphi_n), s \to f
	\end{equation}
	Here $v = \iota_K(\pi(w))$. To prove the ``only if'' direction of the claim we use $(\cdot \to)^{-1}$ in the antecedent of $A_1,\dotsc, A_{m} \to f/s$ and thus release all types of the form $T(\varphi)$, which were combined together within types $A_1,\dotsc,A_{m}$, hence obtaining the sequent of the form $g\times v, T(\varphi_1),\dotsc, T(\varphi_n) \to f/s$. Finally, we move $s$ from the succedent to the antecedent using $(\to /)^{-1}$. The restriction on $n$ arises because there is at most $C$ types of the form $T(\varphi)$ combined within each $A_i$, and the total number of types $A_i$ equals $|w|=m$; therefore, the total number of types is less than or equal to $Cm$.
	
	To prove the ``if'' direction it suffices to find suitable types $A_1$, \dots, $A_m$ such that $w_i \triangleright A_i$ and $\LP\vdash A_1,\dotsc,A_m \to f/s$. Here they are: 
	\begin{enumerate}
		\item $A_i = (g \times \iota_K(\pi(w_i))) \cdot T(\varphi_{j_1}) \cdot \dotsc \cdot T(\varphi_{j_2})$ for $i = 1, \dotsc, \lfloor n/C \rfloor$ where $j_1 = C(i-1)+1$ and $j_2 = Ci$.
		\item $A_i = (g \times \iota_K(\pi(w_i))) \cdot T(\varphi_{j_1}) \cdot \dotsc \cdot T(\varphi_{j_2})$ for $i = \lfloor n/C \rfloor + 1$ where $j_1 = C\lfloor n/C \rfloor + 1$ and $j_2 = n$. It is required that $n/C \not \in \Nat$, otherwise, $A_i = g \times \pi(w_i)$.
		\item $A_i = g \times \iota_K(\pi(w_i))$ for $i > \lfloor n/C \rfloor + 1$.
	\end{enumerate}
	Informally we just distribute all $n$ types $T(\varphi_1),\dotsc,T(\varphi_n)$ among types $A_1,\dotsc,A_m$ in such a way that each $A_i$ contains at most $C$ types of the form $T(\varphi_j)$.
	Note that $g \times \iota_K(\pi(w_i))$ is a single primitive type since $w_i$ is one symbol. According to Construction \ref{construction_lbvassam_to_lp} $w_i \triangleright A_i$; finally, note that $A_1,\dotsc,A_m \to f/s$ is derivable from (\ref{eq_bvass_to_lp}) by using $(\cdot \to)$ several times and $(\to /)$. The claim is proved.
	
	Now it remains to apply Lemma \ref{lemma_main_bvass_to_lp}, which implies that (\ref{eq_bvass_to_lp}) is derivable for some $\varphi_1,\dotsc,\varphi_n \in \mathcal{P}$ if and only if $s(v)$ has a derivation in $G$ of the size $n$. The latter is equivalent to the statement that $\pi(w) \in L(G)$. Summarizing all the steps we obtain that $w \in L(\LPG(G))$ if and only if $\pi(w) \in L(G)$; this is the statement of the theorem.
\end{proof}

\subsection{From LP-Grammars to lBVASSAM}\label{ssec_lp_to_lbvassam}

\begin{construction}[an lBVASSAM corresponding to an $\LP$-grammar]\label{construction_lp_to_lbvassam}
	Assume we are given an $\LP$-grammar $G = \langle \Sigma,S,\triangleright \rangle$, $\Sigma = \{a_1,\dotsc,a_k\}$. We construct an equivalent lBVASSAM $\lBAM(G) \eqdef \langle Q, \mathcal{P}_0,\mathcal{P}_1,\mathcal{P}_2,S,k,K\rangle \vert_{F}$ as follows:
	\begin{enumerate}
		\item $Q \eqdef STp^{+}(G) \cup \{C/A \mid C \in STp^+(G),A \in STp^-(G)\}$.
		\item $K$ equals $|\Sigma|+|STp^{-}(G)|$, and $k$ equals $|\Sigma|$. Hereinafter we fix a bijection $ind:STp^-(G)\to \{k+1,\dotsc,K\}$ (in other words, $ind$ enumerates negative subtypes of $G$ by numbers from $k+1$ up to $K$).
		\item $\mathcal{P}_0$ consists of axioms $p(e_{ind(p)})$ for $p \in Pr \cap STp^-(G)$ (recall that $e_i$ is the $i$-th standard-basis vector in $\Nat^K$ of the form $(0,\dotsc,0,1,0,\dotsc,0)$ where $1$ stands at the $i$-th position).
		\item $\mathcal{P}_1$ consists of the following rules:
		\begin{enumerate}
			\item $C(x+e_{ind(A\cdot B)}) \leftarrow C(x+e_{ind(A)}+e_{ind(B)})$ for $A\cdot B \in STp^-(G)$, $C \in STp^+(G)$;
			\item $(A/B)(x) \leftarrow A(x+e_{ind(B)})$ for $A/B \in Q$;
			\item $A(x+e_{ind(B)}) \leftarrow (A/B)(x)$ for $A/B \in Q$;
			\item\label{rule_4d} $S(x+e_i) \leftarrow S(x+e_{ind(A)})$ for $A$ such that $a_i \triangleright A$.
		\end{enumerate}
		\item $\mathcal{P}_2$ consists of the following rules:
		\begin{enumerate}
			\item $(A \cdot B)(x+y) \leftarrow A(x),B(y)$ for $A\cdot B \in STp^+(G)$;
			\item $(C/(A/B))(x+y) \leftarrow (C/A)(x),B(y)$ for $A/ B \in STp^-(G)$, $C \in STp^+(G)$.
		\end{enumerate}
		\item $S$, which is a distinguished type in $G$, is also a distinguished state in the new grammar.
		\item $F = 6\cdot \max\limits_{A \in Tp(G)} |A|+1$.
	\end{enumerate}
\end{construction}

The main result regarding this construction is the following theorem:
\begin{theorem}\label{th_lp_to_bvass}
	Let $G = \langle \Sigma,S,\triangleright \rangle$ be an $\LP$-grammar. Then $\pi(L(G)) = L(\lBAM(G))$.
\end{theorem}

In the below lemmas we use the same notation as in Construction \ref{construction_lp_to_lbvassam}.

Given a vector $u \in \Nat^K$ such that $u_1 = \dotsc = u_k = 0$, we denote by $\mathcal{A}(u)$ the multiset of types $A_1,\dotsc,A_{|u|}$ such that all $A_i \in STp^-(G)$ and for each $A \in STp^-(G)$ it holds that $|\{i \mid A_i = A\}| = u_{ind(A)}$.
\begin{lemma}\label{lemma_main_lp_to_bvass}
	Let $C \in STp^+(G)$ and let $u \in \Nat^K$ such that $u_1 = \dotsc = u_k = 0$. Then $\LP \vdash \mathcal{A}(u) \to C$ if and only if $C(u)$ has a derivation in $\lBAM(G)$. Besides, in such a case there exists a derivation of $C(u)$ of a size not greater than $3|\mathcal{A}(u) \to C|$.
\end{lemma}
\begin{proof}
	We start with the ``if'' direction. It is proved by induction on the length of a derivation, and the proof is straightforward: the derivation of $C(u)$ can be directly represented as a derivation in $LP$. Let us consider two not so trivial cases of the possible last rule application in the derivation of $C(u)$:
	
	\textbf{Case 1.} The last rule applied is $C(x+e_{ind(B)}) \leftarrow (C/B)(x)$, i.e. $u = u^\prime+e_{ind(B)}$ and $(C/B)(u^\prime)$ is derivable. By the induction hypothesis, the sequent $\mathcal{A}(u^\prime) \to C/B$ is derivable. Using $(\to /)^{-1}$ we obtain that $\LP \vdash \mathcal{A}(u^\prime), B \to C$. It remains to observe that $\mathcal{A}(u^\prime), B = \mathcal{A}(u)$.
	
	\textbf{Case 2.} The last rule applied is $(D/(A/B))(x+y) \leftarrow (D/A)(x),B(y)$, i.e. $C = D/(A/B)$ and $u = u_1+u_2$ for $(D/A)(u_1)$ and $B(u_2)$ being derivable. By the induction hypothesis, sequents $\mathcal{A}(u_1) \to D/A$ and $\mathcal{A}(u_2) \to B$ are derivable. The sequent $\mathcal{A}(u_1),\mathcal{A}(u_2) \to C$ can be derived as follows:
	$$
	\infer[(\mathrm{cut})]{
		\mathcal{A}(u_1),\mathcal{A}(u_2) \to D/(A/B)
	}{	
		\infer[(\to\cdot)]{
			\mathcal{A}(u_1), \mathcal{A}(u_2) \to (D/A)\cdot B
		}{	
			\mathcal{A}(u_1) \to D/A & \mathcal{A}(u_2) \to B
		}
		&
		(D/A)\cdot B \to D/(A/B)
	}
	$$
	It is straightforward to check that $\LP\vdash (D/A)\cdot B \to D/(A/B)$.
	
	The remaining rules are easier to consider because they do not require using the cut rule or invertibility of $(\cdot \to)$ and $(\to /)$. Note that the rule \ref{rule_4d} cannot be applied in the derivation of $C(u)$ because $u_1=\dotsc=u_k=0$.
	
	The other direction is more interesting because we also need to control the size of a derivation. The proof, however, is still straightforward and it is still proved by induction on the length of a derivation (in $\LP$). The base case is where $\mathcal{A}(u) \to C$ is an axiom (say, it is of the form $p \to p$). Then $p(e_{ind(p)})$ is an axiom in $\lBAM(G)$, and hence $u = e_{ind(p)}$ has a derivation of the size $1$.
	
	Now let us prove the induction step in a usual manner by considering possible last rule applications. The only interesting case is the one when $(/\to)$ is applied:
	$$
	\infer[]{
		\Gamma,A/B,\Psi,\Delta \to C
	}
	{
		\Gamma,A,\Delta \to C
		&
		\Psi \to B
	}
	$$
	Here $\Gamma,A/B,\Psi,\Delta = \mathcal{A}(u)$. There exist $u_1$ and $u_2$ such that $\Gamma,A,\Delta = \mathcal{A}(u_1)$ and $\Psi = \mathcal{A}(u_2)$; moreover, we know that $u = u_1+u_2-e_{ind(A)}+e_{ind(A/B)}$. Using the induction hypothesis we conclude that $C(u_1)$ and $B(u_2)$ have derivations of sizes not greater than $3|\mathcal{A}(u_1) \to C|$ and $3|\mathcal{A}(u_2) \to B|$ resp. Finally, we do the following three steps in the derivation:
	\begin{enumerate}
		\item $(C/A)(u_1-e_{ind(A)})\leftarrow C(u_1)$;
		\item $(C/(A/B))(u_1+u_2-e_{ind(A)})\leftarrow (C/A)(u_1-e_{ind(A)}),B(u_2)$;
		\item $C(u_1+u_2-e_{ind(A)}+e_{ind(A/B)}) \leftarrow (C/(A/B))(u_1+u_2-e_{ind(A)})$.
	\end{enumerate}
	This is a derivation of $C(u)$ of the size not greater than $3|\mathcal{A}(u_1) \to C| + 3|\mathcal{A}(u_2) \to B| + 3 = 3|\mathcal{A}(u) \to C|$.
\end{proof}
\begin{lemma}\label{lemma_move_rule_to_end}
	Let $S(u)$ be derivable in $\lBAM(G)$. If there is a step in its derivation of the form $S(v+e_i) \leftarrow S(v+e_{ind(A)})$ for $A$ such that $a_i \triangleright A$, then we can make this step the last one in the derivation.
\end{lemma}
\begin{proof}
	If we do not apply this rule in its original place, then a vector $v+e_{ind(A)}$ appears instead of $v+e_i$ in further steps of the derivation. This does not make the derivation incorrect because there are no rules decreasing the value of the $i$-th component (for $i \le k$). Hence, at the end of the new derivation we obtain $S(u-e_i+e_{ind(A)})$. Applying the rule $S(x+e_i) \leftarrow S(x+e_{ind(A)})$ we complete the derivation.
\end{proof}

Theorem \ref{th_lp_to_bvass} is proved as follows using the above lemmas:

\begin{proof}[Proof (of Theorem \ref{th_lp_to_bvass})]
	Let us prove that, if $u \in L(\lBAM(G))$, then $\pi^{-1}(u) \in L(G)$. A vector $u \in \Nat^k$ belongs to $L(\lBAM(G))$ if and only if $S(v)$ has a derivation of the size not greater than $F|u|$ for $v = \iota_K(u)$. Applying Lemma \ref{lemma_move_rule_to_end} we can assume that rules of the form \ref{rule_4d} are applied after all the remaining rules in the derivation of $S(v)$. In other words, there exists a vector $v^\prime$ such that $v^\prime_1 = \dotsc = v^\prime_k = 0$ and such that the derivation of $S(v)$ is decomposed into two parts:
	\begin{enumerate}
		\item We derive $S(v^\prime)$ without using rules of the form \ref{rule_4d};
		\item We derive $S(v)$ from $S(v^\prime)$ using only rules \ref{rule_4d}. Let us explicitly denote all the steps of this part of the derivation as follows: at the $j$-th step the rule $S(x+e_{i_j}) \leftarrow S(x+e_{ind(A_j)})$ is applied where $a_{i_j} \triangleright A_j$ and $j = 1,\dotsc,|v^\prime|=|v|$.
	\end{enumerate}
	Then proving the fact that $\pi^{-1}(u) \in L(G)$ is straightforward: if we replace each symbol $a_{i_j}$ by a corresponding type $A_j$, then we obtain a multiset of types $A_1,\dotsc,A_{|v|} = \mathcal{A}(v^{\prime})$. Lemma \ref{lemma_main_lp_to_bvass} yields that $\mathcal{A}(v^\prime) \to S$ is derivable, which completes the proof.


	Conversely, let us prove that, if $\pi^{-1}(u) \in L(G)$, then $u \in L(\lBAM(G))$. The former fact implies that there exist such $A_1,\dotsc,A_{|u|}$ that $a_{i_j} \triangleright A_j$ ($j = 1, \dotsc, |u|$) where $\pi^{-1}(u) = a_{i_1},\dotsc, a_{i_{|u|}}$ and such that $\LP \vdash A_1, \dotsc, A_{|u|} \to S$. There exists a vector $\widetilde{v}$ such that $\mathcal{A}(\widetilde{v}) = A_1, \dotsc, A_{|u|}$. Using the ``only if'' direction of Lemma \ref{lemma_main_lp_to_bvass} we conclude that $S(\widetilde{v})$ has a derivation in $\lBAM(G)$ of the size not greater than $3\left|A_1, \dotsc, A_{|u|} \to S\right| = 3\left(\left|A_1\right|+\dotsc+\left|A_{|u|}\right| +\left|S\right| \right) \le 3 \cdot \max\limits_{A \in Tp(G)} |A|\cdot \left(|u|+1\right) \le \left( 6\cdot \max\limits_{A \in Tp(G)} |A| \right) |u|$. Finally, we apply rules of the form \ref{rule_4d} to $S(\widetilde{v})$ $|u|$ times in such a way that the resulting fact is $S(\pi^{-1}(u))$. This finishes the proof.
\end{proof}
Theorems \ref{th_bvass_to_lp} and \ref{th_lp_to_bvass} together can be formulated as follows:
\begin{theorem}\label{th_main}
	A language is generated by an $\LP$-grammar if and only if its Parikh image is generated by a linearly-restricted branching vector addition system with states and additional memory.
\end{theorem}

This allows us to say that $\LP$-grammars and lBVASSAM are equivalent.

\subsection{Corollaries}
Several corollaries of interest follow from Theorems \ref{th_bvass_to_lp} and \ref{th_lp_to_bvass}.
\begin{corollary}\label{cor_quadratic_language}
	There exists an $\LP$-grammar generating the language of multisets $$\{a\times l, b \times n \mid 0<n, 0 \le l \le n^2 \}.$$
\end{corollary}
\begin{proof}
	We present an lBVASSAM $G = \langle Q, \mathcal{P}_0,\mathcal{P}_1,\mathcal{P}_2,s,k,K\rangle \vert_{C}$ generating the set $\{(l,n) \mid 0<n, 0 \le l \le n^2\}$:
	\begin{enumerate}
		\item $K = 7$, $k=2$. For the sake of convenience we rename the standard-basis vectors as follows: $a \eqdef e_1$, $b \eqdef e_2$, $\alpha \eqdef e_3$, $\omega \eqdef e_4$, $\beta \eqdef e_5$, $\rho \eqdef e_6$, $\sigma \eqdef e_7$. 
		\item $Q = \{s\}$.
		\item $\mathcal{P}_0$ includes only $s(\alpha)$.
		\item $\mathcal{P}_1$ consists of unary rules 
		\begin{enumerate}
			\item\label{rule1} $s(x+\alpha+b+\rho) \leftarrow s(x+\alpha)$,
			\item\label{rule2} $s(x+\omega+b+\rho) \leftarrow s(x+\alpha)$,
			\item\label{rule3} $s(x+\omega+\sigma+a) \leftarrow s(x+\omega+b)$,
			\item\label{rule4} $s(x+\beta) \leftarrow s(x+\omega+\rho)$,
			\item\label{rule5} $s(x+\beta+b) \leftarrow s(x+\beta+\sigma)$,
			\item\label{rule6} $s(x+\omega) \leftarrow s(x+\beta)$,
			\item\label{rule7} $s(x) \leftarrow s(x+\omega)$. 
		\end{enumerate}
		\item $\mathcal{P}_2$  is empty.
		\item $C = 4$.
	\end{enumerate}
	Let us say that a derivation in $G$ is \emph{completely typical} if it is of the following form for some $n> 0$:
	\begin{enumerate}
		\item We start with $s(\alpha)$ and apply the rule \ref{rule1} $(n-1)$ times. The result is $$s(\alpha+(n-1)(b+\rho)).$$
		\item We apply \ref{rule2}.
		The result is $$s(\omega+nb+n\rho).$$
		\item The following steps are done for $i=1,\dotsc,n$. Let $x_1 \eqdef 0, y_1 \eqdef n$; then $s(\omega+nb+n\rho) = s(\omega+x_i a+y_i b + (n-y_i)\sigma + (n+1-i)\rho)$.
		\begin{enumerate}
			\item At the beginning of each iteration we have a fact of the form $s(\omega+x_i a+y_i b + (n-y_i)\sigma + (n+1-i)\rho)$.
			\item We apply the rule \ref{rule3} $l_i \le y_i$ times.
			The result is $$s(\omega + (x_i+l_i) a + (y_i-l_i) b+ (n-y_i+l_i)\sigma +(n+1-i)\rho).$$
			\item We apply the rule \ref{rule4}.
			The result is $$s(\beta + (x_i+l_i) a + (y_i-l_i) b+ (n-y_i+l_i)\sigma + (n-i)\rho).$$
			\item We apply the rule \ref{rule5} $l^\prime_i \le n-y_i+l_i$ times.
			The result is $$s(\beta + (x_i+l_i) a + (y_i-l_i+l^\prime_i) b+ (n-y_i+l_i-l^\prime_i)\sigma + (n-i)\rho).$$
			\item We apply the rule \ref{rule6}.
			The result is $$s(\omega + (x_i+l_i) a + (y_i-l_i+l^\prime_i) b+ (n-y_i+l_i-l^\prime_i)\sigma + (n-i)\rho).$$
			This is the last step of the iteration, so $x_{i+1} \eqdef x_i+l_i$, $y_{i+1} \eqdef y_i-l_i+l^\prime_i$.
		\end{enumerate}
		After all these steps being completed we obtain the fact $s(\omega+x_{n+1} a+y_{n+1} b + (n-y_{n+1})\sigma)$.
		\item The rule \ref{rule7} is applied. The result is
		$$s(x_{n+1} a+y_{n+1} b + (n-y_{n+1})\sigma).$$
	\end{enumerate}

	We say that a derivation is \emph{typical} if it is a beginning part of a completely typical derivation. Our claim is that each derivation in $G$ is typical. This is straightforwardly proved by induction on the length of a derivation; the proof is a simple consideration of which rule can be next at each step of a completely typical derivation.
		
	Let $(l,n) \in L(G)$; equivalently, $s(la+nb)$ has a derivation in $G$ of the size less than or equal to $4(n+l)$. This derivation is typical; in fact, it must be completely typical (since both $\omega$ and $\beta$ disappear only at the last step of a completely typical derivation). Hence $l = x_{n+1} = l_1+\dotsc+l_n \le y_1+\dotsc+y_n \le n\cdot n = n^2$ as desired.
	
	Conversely, if $l \le n^2$, then $s(la+nb)$ has a derivation in $G$ of the size $\le 4(n+l)$. Indeed, a derivation of interest is the completely typical one with parameters $l_i = n$ (for $i \le \lfloor l/n \rfloor$), $l_i = l-\lfloor l/n \rfloor n$ (for $i = \lfloor l/n \rfloor + 1$), $l_i = 0$ (for $i > \lfloor l/n \rfloor + 1$); $l^\prime_i = l_i$ (for all $i$). The size of this derivation can be computed by summing the number of rule applications at each stage: it equals $(n-1) + 1 + 2 l + 2 n + 1 = 3n + 2l + 1 \le 4(n+l)$.
\end{proof}
This corollary gives a negative answer to the question of whether $\LP$-grammars generate only permutation closures of context-free languages; indeed, the language $\textsc{perm}\left(\{a^l b^n \mid 0< n, 0 \le l \le n^2\}\right)$ where $\textsc{perm}(L) = \{a_{\sigma(1)}\dotsc a_{\sigma(m)} \mid m>0, a_1\dotsc a_m \in L, \sigma \in S_m\}$ is not a permutation closure of any context-free language (equivalently, it is not a permutation closure of a regular language). This can be proved by, e.g., using Theorem 4 from \cite{DomosiK99} (an iteration lemma for regular languages).

For the Lambek calculus the equivalence of Lambek grammars and context-free grammars allows one to show that Lambek grammars that use only types of a very simple form (namely, either $p$, $p/q$, or $p/q/r$ for $p,q,r \in Pr$) are equivalent to all Lambek grammars; this result can be considered as a normal form for Lambek grammars. In the case of $\LP$-grammars, Constructions \ref{construction_lbvassam_to_lp} and \ref{construction_lp_to_lbvassam} also enable one to establish a normal form for $\LP$-grammars: namely, for each $\LP$-grammar $G$ the grammar $\LPG(\lBAM(G))$ is equivalent to $G$, and it uses only types of a specific form. What interesting result can be proved on the basis of the equivalence of $G$ and $\LPG(\lBAM(G))$? For example, we can slightly modify Construction \ref{construction_lbvassam_to_lp} to prove e.g. the following corollary:
\begin{corollary}
	$\LP$-grammars are equivalent to $\LP(/)$-grammars.
\end{corollary}
\begin{proof}
	The modification of Construction \ref{construction_lbvassam_to_lp} is as follows: 
	\begin{enumerate}
		\item We define a function $T^\prime$ in the same way as $T$ on $\mathcal{P}_0$ and $\mathcal{P}_2$ but we change the definition for $\mathcal{P}_1$: $T^\prime(\varphi_1) \eqdef (f/(f/(p \cdot g^{\delta_1})))/g^{\delta_2}/q$;
		\item We define the grammar $\LPG^\prime(G)$ as $\langle \Sigma, f/s, \triangleright^\prime \rangle$ where $a_i \triangleright^\prime B$ if and only if $B = f/(f/(s \cdot A))/s$ for $A = g_i \cdot T^\prime(\varphi_1)\cdot \dotsc \cdot T^\prime(\varphi_j)$ where $\varphi_l \in \mathcal{P}$ are some axioms and rules and $j \le C$.
	\end{enumerate}
	The proof of the fact that $L(G) = \pi(L(\LPG^\prime(G)))$ is similar to the proof of Theorem \ref{th_bvass_to_lp}. 
	
	Note that $\LP \vdash A/(B\cdot C) \to A/B/C$ and $\LP \vdash A/B/C \to A/(B\cdot C)$. Consequently, in the above grammar, we can replace each type with the product under the division by a type without the product; i.e. we replace $(f/(f/(p \cdot g^{\delta_1})))/g^{\delta_2}/q$ by $(f/(f/p/g^{\delta_1}))/g^{\delta_2}/q$ and $f/(f/(s \cdot A))/s$ by $f/(f/s / g_i / T^\prime(\varphi_1) / \dotsc / T^\prime(\varphi_j))/s$. Hence there exists an $\LP(/)$-grammar equivalent to $\LPG^\prime(G)$. Let us denote it as $\LPG^{\prime\prime}(G)$.
	
	Given an $\LP$-grammar $G$, the grammar $\LPG^{\prime\prime}(\lBAM(G))$ is equivalent to it, and it is an $\LP(/)$-grammar.
\end{proof}

The previous corollary might be formulated even in a stronger way: $\LP$-grammars are equivalent to $\LP(/)$-grammars that use only types of depth less than or equal to $4$. It would be interesting to answer the question if depth can be decreased without losing expressive power or not.

\section{Intersection of languages generated by $\LP$-grammars}\label{sec_intersection}
In this section, we concern with a result concerning languages generated by $\LP$-grammars, which is unrelated to the results from the previous section. It is as follows:

\begin{theorem}
	The class of languages generated by $\LP$-grammars is closed under intersection.
\end{theorem}
The proof of this theorem is similar to that from \cite{Kanazawa92} for grammars over the multiplicative-additive Lambek calculus, and it is quite simple (although we did not find this result in the existing literature). The main idea is that we can use multiplicative conjunction instead of the additive one in the commutative case.
\begin{proof}
	Let $G_i = \langle \Sigma, S_i, \triangleright_i \rangle$ ($i=1,2$) be two $\LP$-grammars; our goal is to find a grammar $G = \langle \Sigma, S, \triangleright \rangle$ such that $L(G) = L(G_1)\cap L(G_2)$. We can assume without loss of generality that $STp(G_1)\cap STp(G_2) = \emptyset$, or, equivalently, that primitive subtypes of types from $G_1$ and $G_2$ are pairwise disjoint. Having this in mind, we define $G$ as follows:
	\begin{itemize}
		\item $S \eqdef S_1\cdot S_2$;
		\item $a \triangleright T$ if and only if $T$ is of the form $T_1\cdot T_2$ where $a \triangleright_i T_i$ ($i=1,2$).
	\end{itemize}
	In the new grammar, $a_1,\dotsc, a_n$ belongs to $L(G)$ if and only if there exist types $T^i_j$ for $i=1,2$, $j=1,\dotsc,n$ such that $a_j \triangleright_i T^i_j$ and $$\LP\vdash T^1_1\cdot T^2_1, \dotsc,T^1_n \cdot T^2_n \to S_1 \cdot S_2.$$ The latter is equivalent to derivability of the sequent $T^1_1, \dotsc, T^1_n, T^2_1, \dotsc, T^2_n \to S_1 \cdot S_2$. 
	
	\begin{lemma}
		\leavevmode
		\begin{enumerate}
			\item Let $\LP \vdash A_1,\dotsc,A_n \to B$ where $A_i$ are from $STp(G)$ and where $B$ is from $STp(G_k)$ ($k\in \{1,2\}$). Then all $A_i$ are also from $STp(G_k)$.
			\item Let $\LP \vdash A_1,\dotsc,A_n,B_1,\dotsc,B_m \to A \cdot B$ where $A_i$ and $A$ are from $STp(G_1)$, and $B_i$ and $B$ are from $STp(G_2)$. Then $\LP \vdash A_1,\dotsc,A_n \to A$ and $\LP \vdash B_1,\dotsc,B_m \to B$.
		\end{enumerate}
	\end{lemma}
	Both statements are proved by a straightforward induction on the length of a derivation. Consequently, $T^1_1, \dotsc, T^1_n, T^2_1, \dotsc, T^2_n \to S_1 \cdot S_2$ is derivable if and only if $\LP \vdash T^1_1, \dotsc, T^1_n \to S_1$ and $\LP \vdash T^2_1, \dotsc, T^2_n \to S_2$. This completes the proof.
\end{proof}

Interestingly, both the class of languages generated by $\LP$-grammars and the class of permutation closures of context-free languages turn out to be closed under intersection; indeed, each permutation closure of a context-free language equals $\pi^{-1}(S)$ for some semilinear set $S$, and semilinear sets are closed under intersection \cite{GinsburgS66}. If this was not the case, we might have a simpler way of proving that $\LP$-grammars are not equivalent to permutation closures of context-free languages.

\section{Conclusion}\label{sec_conclusion}
We showed that $\LP$-grammars are not context-free in the sense that they generate more than permutation closures of context-free languages. This result contrasts with that for Lambek grammars, which are context free \cite{Pentus93}. We proved this by establishing the equivalence of $\LP$-grammars and lBVASSAM, which is yet another extension of vector addition systems. 

Several open questions related to the achieved results could be mentioned:
\begin{enumerate}
	\item Is the set of languages generated by $\LP$-grammars closed under complement? Note that semilinear sets are. If the answer to this question is negative, then this would give us another proof of the fact that languages of $\LP$-grammars are more than permutation closures of context-free languages.
	\item Can we generate a language like $\{a\times f(n) \mid n>0\}$ where $f(n)$ is some nonlinear function (e.g. $f(n)=n^2$) by an $\LP$-grammar?
	\item The linear restriction can be imposed on BVASS as well resulting in lBVASS. Then one might ask whether lBVASS are equivalent to lBVASSAM, i.e. whether additional memory is essential in lBVASSAM. Similarly, it would be interesting to answer the question if BVASS are equivalent to BVASSAM. 
	\item Generally, we are not completely aware that there is no other formalism existing in the literature, which would appear to be equivalent to (l)BVASSAM (that is, that we cannot rid of our own definitions).
\end{enumerate}

Although the methods used in this paper do not exploit high-level technical tricks but they are rather straightforward, we would like to emphasize the importance of the \emph{linear restriction}. It naturally arises for $\LP$-grammars, and, moreover, it proved to be useful as well for hypergraph Lambek grammars: namely, in \cite{Pshenitsyn22}, we prove that any DPO hypergraph grammar (DPO grammars extend unrestricted Chomsky grammars to hypergraphs) can be transformed into an equivalent hypergraph Lambek grammars. The main construction in that paper is based on the same idea as Construction \ref{construction_lbvassam_to_lp}. We assume that using this restriction for other kinds of formal grammars can also be successfully used for investigating expressive power of other kinds of categorial grammars.

\bmhead{Acknowledgments}
I thank Stepan L. Kuznetsov for bringing my attention to this problem and for suggesting valuable ideas to explore.

\section*{Declarations}

The study was supported by RFBR, project number 20-01-00670, by the Theoretical Physics and Mathematics Advancement Foundation ``BASIS'', and by the Interdisciplinary Scientific and Educational School of Moscow University ``Brain, Cognitive Systems, Artificial Intelligence''.

\begin{appendices}



\end{appendices}


\bibliography{LP-Grammars}


\end{document}